\documentclass[twocolumn,10pt,final]{IEEEtran}

\input{def.tex}
\usepackage{dsfont}

\DeclareSymbolFont{matha}{OML}{txmi}{m}{it}
\DeclareMathSymbol{\varv}{\mathord}{matha}{118}

\makeatletter
\DeclareUrlCommand\ULurl@@{%
  \def\UrlLeft{\uline\bgroup}%
  \def\UrlRight{\egroup}}
\def\ULurl@#1{\hyper@linkurl{\ULurl@@{#1}}{#1}}
\DeclareRobustCommand*\ULurl{\hyper@normalise\ULurl@}
\renewcommand{\baselinestretch}{0.915}

\begin{document}
\title{Downlink MIMO HCNs with Residual Transceiver Hardware Impairments}
\author{Anastasios Papazafeiropoulos and Tharm Ratnarajah,   \vspace{2mm} \\
\thanks{A. Papazafeiropoulos and T. Ratnarajah are  with the  Institute for Digital Communications (IDCOM), University of Edinburgh, Edinburgh, EH9 3JL, U.K., (email: {a.papazafeiropoulos, t.ratnarajah}@ed.ac.uk). }
\thanks{This work was supported by the U.K. Engineering and Physical Sciences Research Council (EPSRC) under grant EP/L025299/1.}}\maketitle
%

\begin{abstract}
A major limitation of heterogeneous cellular networks (HCNs) is the neglect of the additive residual transceiver hardware impairments (ARTHIs). The assumption of perfect  hardware is quite strong and results in misleading conclusions. This paper models a general multiple-input multiple-output (MIMO) HCN with cell association by incorporating the RTHIs. We derive the coverage probability and shed light on the impact of the ARTHIs, when  various transmission methods are applied. As the hardware quality decreases, the coverage probability worsens. Especially, this effect is more severe as the transmit power increases. Furthermore, we verify that in an HCN, it is better to employ at each base station as few transmit antennas as possible.
\end{abstract}

\section{Introduction}
Recently, heterogeneous cellular networks (HCNs) have attracted a significant interest for 5th generation (5G) wireless systems~\cite{Osseiran2014}. The maturity of HCNs, started from single-input single-output (SISO) links~\cite{Andrews2011}, has enabled the coexistence of multiple-antenna strategies~\cite{Kountouris2012,Dhillon2013}.

Over the years, there has been an increasing focus on investigating the effects of transceiver hardware impairments (THIs) on the performance of wireless communication systems such as phase noise~\cite{Papazafeiropoulos2016}, high power amplifier nonlinearities~\cite{Qi2012}, In-phase/Quadrature-phase (I/Q)-imbalance~\cite{Qi2010}. Although calibrations schemes at the transmitter and compensation algorithms at the receiver  exist, their efficacy is limited, since  a certain amount of inevitable residual impairments still remains due to several reasons, e.g., the time-variation of the hardware characteristics.  Thus, the additive residual THIs (RTHIs), modeling the aggregate effect of all the residual transceiver impairments, arise~\cite{Studer2010,Bjoernson2013,Bjornson2014,Bjornson2015,Papazafeiropoulos2015b,PapazafeiropoulosMay2016}. Despite that HCNs are a candidate solution for 5G  systems~\cite{Osseiran2014}, no evaluation of the impact of the ARTHIs has taken place regarding HCNs in the literature.

In this paper,  we make a step beyond~\cite{Kountouris2012} and~\cite{Dhillon2013}, which considered ideal hardware, in order to assess the effect of the ARTHIs on HCNs. Specifically, we consider a downlink multiple-input multiple-output (MIMO) HCN in the presence of the RTHIs. Moreover, we derive the coverage probability in terms of tools from stochastic geometry. The result enables us to illustrate the impact of the ARTHIs on the performance of HCNs and draw a picture on their effects on the multiple-antenna transmission strategies.

\section{System Model}
In this paper, we consider a cellular MU-MISO system, having one BS per cell, drawn according to an independent Poisson Point Process (PPP) $\Phi_{B}$ with density $\lambda_{B}$. 
Each BS deploys a large number of antennas $M$ that is greater or equal to the number of associated users $K$, i.e., $M\ge K$. Also, the user locations are modeled by an independent PPP $\Phi_{u}$ with density $\lambda_{u}=6 \lambda_{B}$.  Moreover, the same time-frequency resource is shared by the users across all cells. Slivnyak's theorem allows conducting the analysis by focusing on a typical user found at the origin. 
We assume that the  users belong to the Voronoi cell of the nearest BS, and the set of all the cells comprise a Voronoi tessellation.
\subsection{Downlink Transmission}
We assume that the desired channel power from the BS located at $x\in \mathbb{R}^{2}$ to the typical user, found in its cell, is given by $h_{k}$, while the inter-cell interference power from another BS (located at $y_{l} \in \mathbb{R}^{2}$) is denoted by $g_{y}$. 

Assuming knowledge of perfect channel state information (CSI) at the transmitter side, we focus on ZF precoding. Hence, the received signal from the $j$th BS  to user $k$ at $x$ in its cell, after applying the ZF precoder, can be expressed as
\begin{align}
 y_{k}=\bh_{k}^{\H}  \bs_{k} \|x \|^{-\al/2}
 +\!\!\!\sum_{l\in \Phi_{B}/x}\!\!\!\! \bg_{l}^{\H}  \bs_{l} \|y_{l} \|^{-\al/2}+n_{k}, \label{signal1} 
\end{align}
where   $\bs_{k} =\bW_{k} \bd_{k} \in \mathbb{C}^{M \times 1}$ is the  transmit signal vector for the $k$th user with covariance matrix $\bQ=\EE\left[ \bs_{k}\bs_{k}^{\H}\right] $ and   $p_{k}=\tr\left( \bQ \right)$ is the associated average transmit power. In particular, we assume that the  linear precoding is denoted by the matrix $\bW_{k}\in \mathbb{C}^{M\times K}$, employed by the  BS, which multiplies the data signal vector $\bd_{k} = \big[d_{k,1},~d_{k,2},\cdots,~d_{k,K}\big]^\T \in \mathbb{C}^{K}\sim \mathcal{CN}(\b0,\bI_{K})$ for all users in that cell.  Also, $\alpha$ is the path-loss exponent parameter. The channel vectors $\bh_{k} \in \mathbb{C}^{M\times 1}$ and $\bg_{l} \in \mathbb{C}^{M\times 1}$ denote the desired and interference channel vectors between BSs located at $x\in \mathbb{R}^{2}$ and $y_{l}\in \mathbb{R}^{2}$ far from the typical user.  In the case of Rayleigh fading, the channel power distributions of both the direct and the interfering links follow the Gamma distribution~\cite{Huang2011}.  Also,  $\bn_{k} $ is an additive white Gaussian noise (AWGN) vector, such that $n_{k} \sim \mathcal{CN}\left( 0, 1\right)$.  

In practice, both the users and the BSs are affected by certain inevitable residual additive  impairments~\cite{Schenk2008,Studer2010}. 
Given the channel realizations, the conditional transmitter and  receiver  distortion noises for the $i$th link are modeled as Gaussian distributed, where their average power is proportional to the average signal power, as shown by measurement results~\cite{Studer2010}. In other words, we have\footnote{We assume that all the BSs have the same hardware impairments without any loss of generality.}
\begin{align} 
 \etv_{\mathrm{t}}&\sim \cC\cN(\b0,\delta_{\mathrm{t}}^{2}\mathrm{diag}\left( q_{1},\ldots,q_{M} \right)),\\
 \eta_{\mathrm{r}}&\sim \cC\cN(\b0,\delta_{\mathrm{r}}^{2}\|x\|^{-\al}\bh^{\H}_{k}\tr\!\left( \bQ \right)\bh_{k})
\end{align}
where  $q_{1},\ldots,q_{M}$ are the diagonal elements of $\bQ$. Note that the circularly-symmetric complex Gaussianity can be justified by the aggregate contribution of many impairments. The proportionality parameters $\delta_{\mathrm{t}}^{2}$ and $\delta_{\mathrm{r}}^{2}$   describe the severity of the residual impairments at the transmitter and the receiver side. In applications, these parameters are met as the error vector magnitudes (EVM) at each transceiver side~\cite{Holma2011}.
\begin{remark}
 The receive distortion includes the path-loss coming from the associated BS\footnote{Note that the ARTHIs
 from other BSs are negligible due to the increased path-loss. Also,  the transmit hardware impairment depends only on the transmit signal power from the tagged BS and not from the path-loss.}.
\end{remark}

Hence, the hardware impairments are written as
$  \etv_{\mathrm{t}}\sim \cC\cN(\b0,{p_{k}\delta_{\mathrm{t}}^{2}}  \Id_{K})$  and  $\eta_{\mathrm{r}}\sim \cC\cN(\b0,p_{k}{\delta_{\mathrm{r}}^{2}}\|x\|^{-\al}\|\bh_{k}\|^{2} )$.

Incorporating these parameters to~\eqref{signal1}, we obtain
\begin{align}
  y_{k}\!=\!p_{k}\bh_{k}^{\H} \! \left( \bs_{k}\!+ \! \etv_{\mathrm{t}}\right)\! \|x \|^{\!-\!\frac{\al}{2}}\!\!+\!\!\!\!\sum_{l\in \Phi_{B}/x}\!\!\!\!p_{l}\bg_{l}^{\H}  \bs_{l} \|y_{l} \|^{\!-\!\frac{\al}{2}}\!+\!\eta_{\mathrm{r}}\!+\!n_{k}\nn.
\end{align}

\begin{proposition}\label{SINR}
 The  signal-to-interference (SIR) ratio of the downlink transmission from the BS to  its typical user, accounting for transceiver hardware impairments, in an HCN can be represented  by
 \begin{align}
    \gamma_{k} \label{eq Theo 1}
    =
                \frac{
               p_{k}  h_{k} \|x \|^{-\al} 
                    }{
                   I_{ \etv_{\mathrm{t}}}\|x \|^{-\al} +  I_{l}+I_{ \etv_{\mathrm{r}}}\|x\|^{-\al}
                    },
 \end{align}
 where $h_{k}=|\bh_{k}\bw_{k,k}|^{2}\sim
       \Gamma\left( \Delta_{k},1 \right),  \mathrm{and}~  \Delta_{k}=M-K+1$. Note that $\bw_{k,k}$ is the $k$th column of $\bW_{k}$.
In addition, we have $
I_{ \etv_{\mathrm{t}}}\sim {p_{k}\delta_{\mathrm{t}}^{2}}\Gamma (M,1)$,  and $
I_{ \etv_{\mathrm{r}}}\sim  p_{k}{\delta_{\mathrm{r}}^{2}}\Gamma (M,1)$
while the total interference from all the other base stations found at a distance $\|y\|$ from the typical user 
is
 $I_{l} 
           = \sum_{l \in \Phi_{B}/x} p_{l}g_{l}\|y \|^{-\al}$, 
            where $g_{l}\sim \Gamma (K,1)$.
\end{proposition}

\begin{proof}
See Appendix~\ref{SINRproof}.
\end{proof}
\section{Coverage Probability}
This section, starting with a formal definition of the coverage probability with BS locations drawn from a PPP, presents the technical derivation  of an upper bound of the downlink coverage probability of a typical user in a MIMO HCN. While different transmission techniques are employed that depend on the number of BS antennas $M$ and the number of users in each cell $K$, the inherent existence of residual additive transceiver hardware impairments is incorporated in the analysis. 

{The generality of the model allows the investigation of the effects of hardware imperfections on the coverage probability towards  a more realistic  assessment.}
\begin{definition}[\!\!\cite{Dhillon2013}]
A typical user  is  in  coverage if its effective  downlink SIR from at least one of the randomly located BSs in the network is higher
  than the corresponding target. In general, we have\footnote{We assume that the thermal noise is negligible as compared to the distortion noises and the other cells interference as showed by simulations. However, it can be included in the proposed analysis by means of some extra work.} 
\begin{align}
 p_{c}=\EE\left[\mathds{1}\!\left( 	\underset{x \in \Phi_{B}}{\cup} \mathrm{SIR\left( x \right)>T} \right)  \right],
\end{align}
where  the indicator function $ \mathds{1}(e)$ is $1$ when event $e$ holds and $0$ otherwise. 
\end{definition}

The following theorem is the main result, being unique in the research area of practical systems with hardware impairments, when the BSs are randomly positioned. It is based on the calculation of the Laplace transforms provided by means of Proposition~\ref{LaplaceTransform} and Lemma~\ref{LaplaceTransformGamma}. 
\begin{theorem}\label{theoremCoverageProbability} 
The downlink probability of coverage  $p_{c}\left( T,\lambda_{B},\alpha,\delta_{\mathrm{t}},\delta_{\mathrm{r}}  \right)$ in a general cellular network with randomly distributed multiple-antenna BSs, accounting for additive transceiver hardware impairments, is given by
 \begin{align}
&\!\!p_{c}\left(  T,\lambda_{B},\alpha,\delta_{\mathrm{t}},\delta_{\mathrm{r}} \right)\!\le\!\!\lambda_{B}\int_{l \in  \mathds{R}^2}\!\!
                    \sum_{i=0}^{\Delta-1}\sum_{k=0}^{i}\sum_{n=0}^{i-k}\binom{i}{k}\binom{i-k}{n}\nn\\
                  &\!\!\times\!\frac{\left( \!-\!1 \right)^{\!i}\!\tilde{T}^{i-k}s^{k}  }{i!}\frac{\mathrm{d}^{n}}{\mathrm{d}s^{n}}\mathcal{L}_{I_{ \etv_{\mathrm{r}}}}\!\!\left(s \right)\!                 
                  \frac{\mathrm{d}^{i-k-n}}{\mathrm{d}s^{i-k-n}}\mathcal{L}_{I_{ \etv_{\mathrm{t}}}}\!\!\left(s \right)\!\!\frac{\mathrm{d}^{k}}{\mathrm{d}s^{k}} \mathcal{L}_{I_{l}}\!\left(s \right)\mathrm{d}l,\label{coverageprobability} \!
\end{align}
where  $l=\|x\|$,  $s=\tilde{T}l^{a}$,  and  $\tilde{T}=T p_{k}^{-1}$, while $\mathcal{L}_{I_{ \etv_{\mathrm{r}}}}\!\left(s \right)$,   $\mathcal{L}_{I_{ \etv_{\mathrm{t}}}}\!\left(s \right)$, and $\mathcal{L}_{I_{l}}\!\left(s \right)$ are the Laplace transforms of the powers of the receive distortion, transmit distortion, and  interference power coming from other BSs.
\end{theorem}
\begin{proof}
See Appendix~\ref{CoverageProbabilityproof}.
\end{proof}
\begin{remark}
In the ideal case of no  ARTHIs,~\eqref{coverageprobability} coincides with the coverage probability provided by Theorem~$3$ in~\cite{Dhillon2013} for single tier.
\end{remark}

\begin{proposition}\label{LaplaceTransform} 
The Laplace transform of the interference power of a  general cellular network with randomly distributed multiple-antenna BSs having additive transceiver hardware impairments is given by
\begin{align}
&\mathcal{L}_{I_{l}}\!\left({s} \right)=\exp{\!\left( - {s}^{\frac{2}{a}} \mathcal{C}\left( \al, K\right) \right)},
\end{align}
where  $\mathcal{C}\left( \al, K\right)=\frac{2 \pi \lambda_{B}}{a} \sum_{m=0}^{K}\binom{K}{m}\mathrm{B}\left( K-m+\frac{2}{a},m-\frac{2}{a} \right)$.
\end{proposition}
\begin{proof}
See Appendix~\ref{LaplaceTransformproof}.
\end{proof}
\begin{lemma}\label{LaplaceTransformGamma} 
 The Laplace transforms of the parts, describing the ARTHIs $I_{ \etv_{\mathrm{t}}}$ and $I_{ \etv_{\mathrm{r}}}$, are given by
 \begin{align}
\mathcal{L}_{I_{ \etv_{j}}}\!\left(s \right)&=\frac{1}{\left( 1+{q_{j}}s \right)^{M}}, 
 \end{align}
where $j=\mathrm{t}$ or $\mathrm{r}$ and $q_{t}=\delta_{\mathrm{t}}^{2}$ or  $q_{r}=\delta_{\mathrm{r}}^{2}$, respectively.
\end{lemma}
\begin{proof}
Both  Laplace transforms are easily obtained, since $I_{ \etv_{\mathrm{t}}}$ and $I_{ \etv_{\mathrm{r}}}$ follow scaled gamma distributions with scaled parameters ${p_{k}\delta_{\mathrm{t}}^{2}}$ and $p_{k}{\delta_{\mathrm{r}}^{2}}$, as mentioned in Appendix~\ref{SINRproof}.
\end{proof}
 
The numerical evaluation of~\eqref{coverageprobability} is complex and time-consuming because it involves the calculation of the derivatives of Laplace transforms. 
\begin{remark}
In the general case, where $\Delta>1$, the derivatives of the Laplace transforms, being composite functions,  are calculated by applying  Fa\`{a} di Bruno's identity. In the case of $\mathcal{L}_{I_{l}}\!\left(s \right)$, if we denote the composite function as $\left( f \circ g \right)\left( s \right)$, then $f\left( s \right)=\exp\left( s \right)$ and $g\left( s \right)=- \tilde{s}^{\frac{2}{a}} \mathcal{C}\left( \al, \mathcal{M}\right)$. Similarly, in the case of $\mathcal{L}_{I_{ \etv_{j}}}\!\left(s \right)$, we have $f\left( s \right)=\frac{1}{s^{M}}$ and $g\left(s \right)=1+{q_{j}^{2}}s$.
\end{remark}
\begin{corollary}
In the special case of full  space
division multiple access (SDMA) $(M=K)$, the upper bound of the coverage probability with residual transceiver impairments, described by Theorem~\ref{theoremCoverageProbability}, is given by
\begin{align}
\! \!\!p_{c}\left(  T,\lambda_{B},\alpha,\delta_{\mathrm{t}},\delta_{\mathrm{r}} \right)\!\le\!\!\lambda_{B}\!\! \int_{x \in  \mathds{R}^2}\!\!\mathcal{L}_{I_{ \etv_{\mathrm{r}}}}\!\!\left(s \right)                  
                 \mathcal{L}_{I_{ \etv_{\mathrm{t}}}}\!\!\left(s \right)
                   \mathcal{L}_{I_{l}}\!\!\left(s \right)\mathrm{d}x.\!\!
\end{align}
\end{corollary}
\section{Numerical Results}
The locations of  BSs are simulated as realizations of a PPP with given density in a  sufficiently large window of $5~ \mathrm{km}~\times 6~\mathrm{km}$. Moreover, we assume that the typical user lies at the origin, and we 
calculate the desired signal strength and the interference power. The coverage probability is obtained by checking if the received SIR from at least one of the BSs is more than the target value.  The ``solid''  and ``dot'' lines designate the analytical results with no  ARTHIs and specific  ARTHIs, respectively, while the bullets represent the simulation results. 

In Fig.~$1$, the simulated coverage probability $p_{c}$ along with the proposed analytical result~\eqref{coverageprobability} are plotted against the target SIR $T$ for $\delta_{t}=0.15$ and $\delta_{r}=0$\footnote{Note that based on the proposed model the transmit and receive additive impairments have equivalent effect, although it should be stressed that in reality the BS's transmitter and user's receiver are manufactured with different quality.}. These nominal values of  ARTHIs are quite reasonable according to~\cite{Bjornson2014}. Moreover, in the same figure, we have depicted the simulated and theoretical results corresponding to ideal hardware as provided by (21) in~\cite{Dhillon2013}. Obviously, in practice, where  ARTHIs exist, $p_{c}$ worsens as lower hardware quality is used (increasing $\delta_{t}, \delta_{r}$).  In addition, we consider three different transmission strategies. In particular, we have i) single-user beamforming (SU-BF) with $M=6$, $K=1$, ii) $M=K=1$, i.e, each BS has a  single transmit antenna (SISO), and iii) $M=K=6$, which means  full SDMA. Similar to~\cite{Dhillon2013}\footnote{In interference-limited networks, the claim having more antennas is always beneficial is not necessarily correct, as it heavily depends on how the transmit antennas are used and which transmission/reception scheme is employed. For instance, using the transmit antennas for multi-stream transmission (SDMA) is not beneficial (from a system perspective) in most cases (with treating interference as noise). Furthermore, the  correctness of the claim depends on the performance metric we study. In general, there are regimes/situations where SU-BF is better than SISO. The main reason why SU-BF can outperform SISO is that in addition to the proximity gain enjoyed by the SISO due to extreme densification, the SU-BF transmission presents an  additional beamforming gain. Also,   we should take into account that the growth of the received signal and interference for increasing $\lambda_{B}$  is the same.}, we show that  SU-BF transmission is preferable with comparison to SISO, while the latter is better than SDMA. In other words, we verify that it is better to serve a single user in each resource block, either by SISO or SU-BF, instead of serving multiple users. However, herein and with comparison to~\cite{Dhillon2013}, we illustrate how, given this property, the coverage probability varies with SIR in the presence of the ARTHIs.

SU-BF case has an additional beamforming gain;

Fig.~2 further illustrates the effect of increasing the transmit BS power. This exposes a quite insightful property, since the ARTHIs are  power-dependent. Specifically, increasing the ARTHIs, the space $l_{i}$ for $i=1,~2$, representing the gap between the lines with no hardware impairments and the  practical scenario with  ARTHIs, becomes larger, i.e., $l_{1}>l_{2}$ for $\rho=15~\mathrm{dB}$ and $\rho=5~\mathrm{dB}$, respectively.
\begin{figure}[!h]
 \begin{center}
  \includegraphics[width=0.98\linewidth]{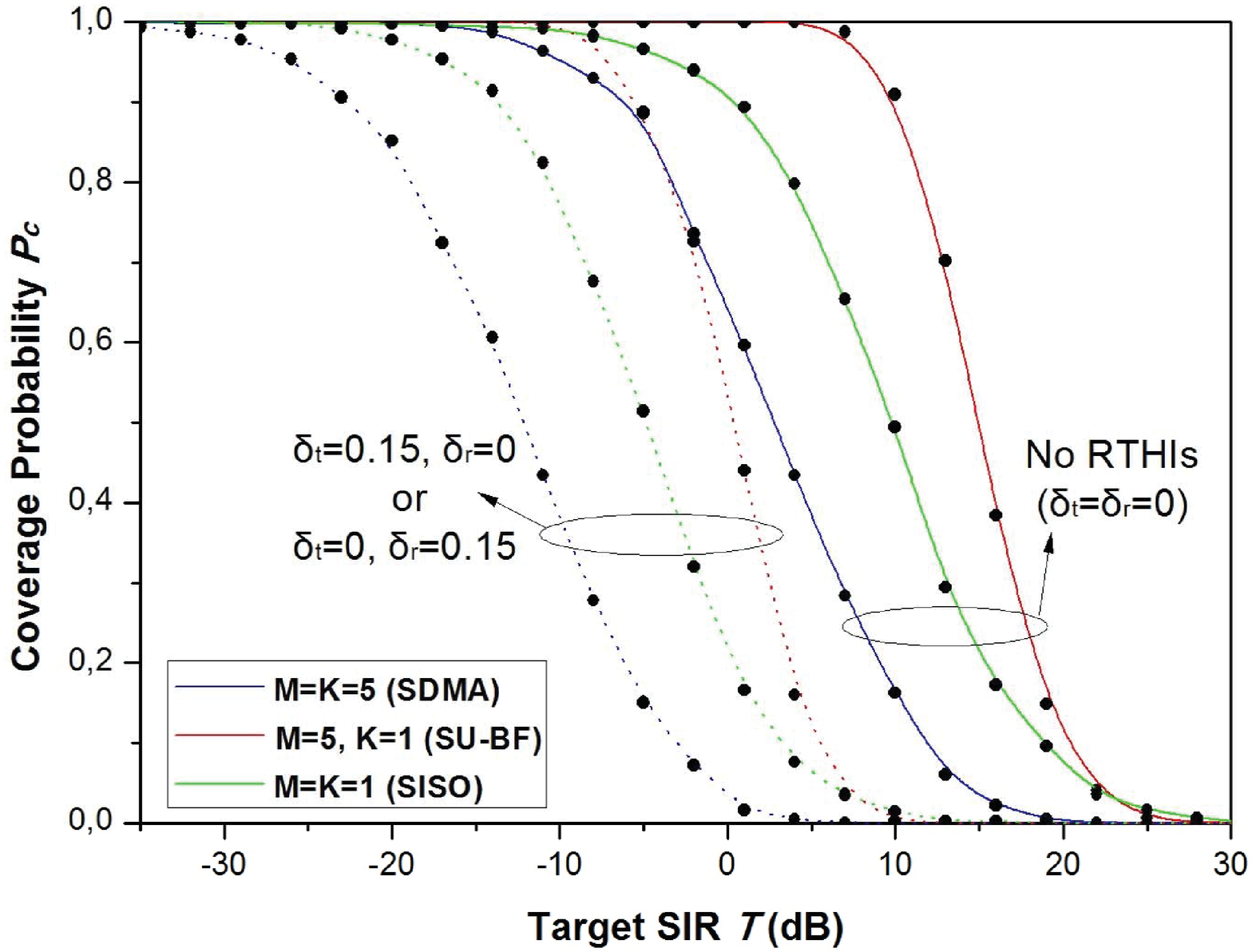}
 \caption{\footnotesize{Coverage probability versus the target SIR $T$ for varying  ARTHIs and various   transmission techniques ($\al=3,\lambda_{B}=3 $, $\lambda_{u}=6 \lambda_{B}$,  $p= 23~\mathrm{dB}$.)}}
 \label{M=100}
 \end{center}
 \end{figure}\begin{figure}[!h]
 \begin{center}
 \includegraphics[width=0.98\linewidth]{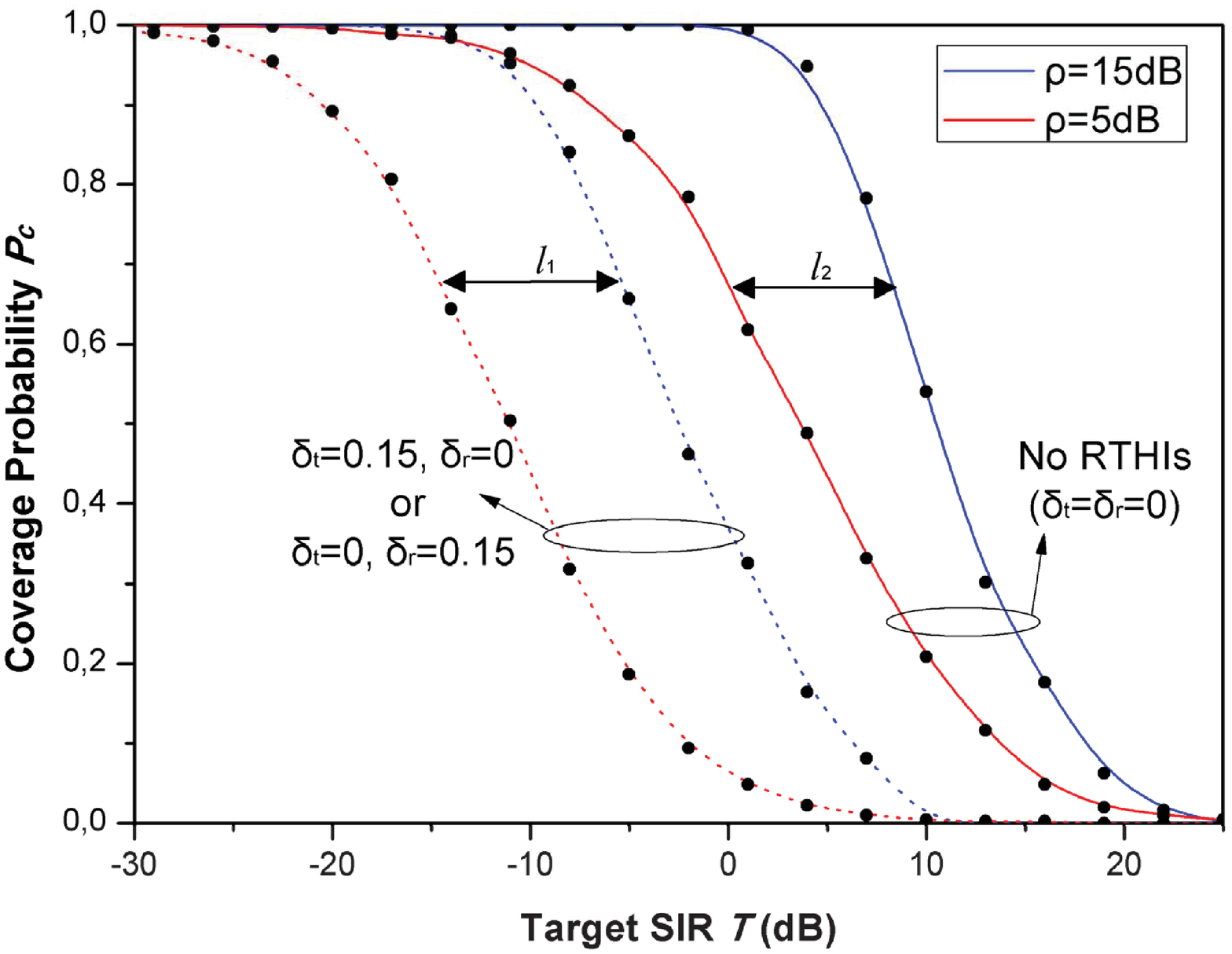}
 \caption{\footnotesize{Coverage probability versus the target SIR $T$ for varying  ARTHIs and various   transmit powers ($\al=3,\lambda_{B}=3 $, $\lambda_{u}=6 \lambda_{B}$.)}}
 \label{M=100}
 \end{center}
 \end{figure}
\section{Conclusion}
Contrary to existing works on HCNs, assuming ideal hardware,  this paper studied the impact of the ARTHIs. In particular, based on a general realistic scenario, where a BS can employ several multiple-antennas downlink transmission strategies after taking into account for cell association, we obtained the coverage probability  in the presence of the unavoidable ARTHIs. It was showed that the ARTHIs degrade the coverage capability. More importantly, it was showed that this degradation is higher as the transmit power increases.


\begin{appendices}
\section{Proof of Proposition~\ref{SINR}}\label{SINRproof}
We assume that the columns of the precoding matrix $\bW_{k}$ equal the normalized columns of $\bH^{\H}\left( \bH \bH^{\H} \right)^{-1}$, i.e., $\bW_{k}=\bar{\bH}^{\H}\left( \bar{\bH} \bar{\bH}^{\H} \right)^{-1}$, where $\bar{\bH}=\left[ \bar{\bh}_{1},\ldots, \bar{\bh}_{k}\right]^{\H}\in \mathbb{C}^{\left( K \times M \right)}$ with columns $\bar{\bh}_{k}=\frac{\bh_{k}}{\|\bh_{k}\|}$. In such case, the desired signal power, given by $h_{k}=|\bar{\bh}_{k}^{\H}\bw_{k,k}|\cdot \|\bh_{k}\|^{2}$, is $\Gamma\left( \Delta,1 \right)$ distributed with $\Delta=M-K+1$, since it equals to the product of two independent random variables distributed as $B\left( M-K+1,K-1 \right)$ and $\Gamma\left( M,1 \right)$, respectively~\cite{Jindal2006}.  $I_{ \etv_{\mathrm{t}}}$ is obtained after taking the expectation over the transmit distorion noise of the tagged BS  $I_{ \etv_{\mathrm{t}}}={p_{k}\delta_{\mathrm{t}}^{2}}\|\bh_{k}\|^2$, which follows a scaled $\Gamma (M,1)$ distribution. A similar result is obtained after taking the expectation over the receive distorion noise, i.e., $I_{ \etv_{\mathrm{r}}}=p_{k}{\delta_{\mathrm{r}}^{2}}\|\bh_{k}\|^2$, The other term in the denominator, concerning the interference from other BSs, $I_{l}$, is expressed in terms of the sum of two independent gamma distributed random variables $g_{l}=|\bg_{l}^{\H}\bs_{l}|^2\sim \Gamma (K,1)$. Note that $g_{l}$ is a  $\Gamma (K,1)$ random variable because the precoding matrices $\bW_{l}$ coming from other BSs have  unit-norm and are independent from the normalized $\bar{\bg}_{l}$. Therefore, $g_{l}=\bar{\bg}^{\H}_{l}\bW_{l}$  is a linear combination of $K$ complex normal random variables, i.e., $g_{l}\sim \Gamma(K,1)$.

\section{Proof of Theorem~\ref{theoremCoverageProbability}}\label{CoverageProbabilityproof}
According to the definition of $p_{c}\left(  T,\lambda_{B},\alpha,\delta_{\mathrm{t}},\delta_{\mathrm{r}} \right)$ and by means of appropriate substitution of the SIR $\gamma_{k}$, we have
\begin{align}
 \!\!&p_{c}\left(  T,\lambda_{B},\alpha,\delta_{\mathrm{t}},\delta_{\mathrm{r}} \right)=\EE\left[\mathds{1}\!\left( 	\underset{x \in \Phi_{B}}{\cup} \mathrm{SIR}\left( x \right)>T \right)  \right]\\
 &\le\EE\left[ 	\underset{x \in \Phi_{B}}{\cup} \mathds{1}\left(  \mathrm{SIR} \right)>T   \right]\\
\! \!&=\!\!\EE \left[\sum_{x \in \Phi_{B}}\!\!\mathbb{P}\left[ \mathrm{SIR}>T|l \right]\right]\label{coverage_definition}\\
                                   \! \!&=\!\!  \lambda_{B}\int_{x \in  \mathds{R}^2}\!\!\EE\left[ \mathbb{P}\left[   h_{k}>
                                  \tilde{T}  \left( I_{ \etv_{\mathrm{t}}}+I_{ \etv_{\mathrm{r}}} \right) +  \tilde{T} l^{\al} I_{l}|l \right]\right]\mathrm{d}x,\label{coverage_definition1} 
\end{align}
where in~\eqref{coverage_definition}, we have used the Campbell-Mecke Theorem [30]. Given that $h_{k}$ is gamma distributed, i.e., $ h_{k} \overset{\tt d}{\sim}
   \Gamma\left( \Delta_{k},1 \right)$, we have 
     $\mathbb{P}_{h_{k}}\left( z \right) =e^{-z}\sum_{i=0}^{\Delta-1}\frac{z^{i}}{i!}$.
Thus, the integrable part of~\eqref{coverage_definition1} can be written as
\begin{align}
 &\mathbb{P}\!\left[h_{k}\!>\!\tilde{T}  \left( I_{ \etv_{\mathrm{t}}}+I_{ \etv_{\mathrm{r}}} \right) +  \tilde{T} l^{\al} I_{l}|l \right]=e^{-\left( \tilde{T} \left( I_{ \etv_{\mathrm{t}}}+I_{ \etv_{\mathrm{r}}} \right) +  \tilde{T} l^{\al} I_{l} \right)}
\nn\\
&\times\sum_{i=0}^{\Delta-1}\sum_{k=0}^{i}\binom{i}{k}\frac{\left(  \tilde{T}  \left( I_{ \etv_{\mathrm{t}}}+I_{ \etv_{\mathrm{r}}} \right)  \right)^{i-k}\left(     \tilde{T} l^{\al} I_{l}  \right)^{k}}{i!},\label{coverage1}
                                \end{align}
                                where in~\eqref{coverage1}, we have applied the binomial theorem. Taking the expectation, we obtain 
\begin{align}
 &\!\EE\!\left[\!\mathbb{P}\!\left[h_{k}\!>\!\tilde{T}\! \left( I_{ \etv_{\mathrm{t}}}\!+\!I_{ \etv_{\mathrm{r}}} \right)\! + \! \tilde{T} l^{\al} I_{l}|l \right]\right]\!=\!\!
                    \sum_{i=0}^{\Delta-1}\!\sum_{k=0}^{i}\!\sum_{n=0}^{i-k}\!\!\binom{i}{k}\!\binom{i\!-\!k}{n}\nn\\
                  &\!\times\!\frac{\left( -1 \right)^{i}\!\tilde{T}^{i-k}s^{k}  }{i!}\frac{\mathrm{d}^{n}}{\mathrm{d}s^{n}}\mathcal{L}_{I_{ \etv_{\mathrm{r}}}}\!\!\left(s \right)\!                  
                  \frac{\mathrm{d}^{i-k-n}}{\mathrm{d}s^{i-k-n}}\mathcal{L}_{I_{ \etv_{\mathrm{t}}}}\!\!\left(s \right)\!\frac{\mathrm{d}^{k}}{\mathrm{d}s^{k}} \mathcal{L}_{I_{l}}\!\!\left(s \right)\!,\!\label{coverage3}
\end{align}
where we have set $\tilde{T}=T p_{k}^{-1}$ and $ s=\tilde{T} l^{\al}$. In~\eqref{coverage1}  we have made use of the Binomial theorem, and in~\eqref{coverage3} we have used the definition of the Laplace Transform $\mathbb{E}_{I }\left[  e^{-s I }\left( s I  \right)^{i}\right]=s^{i}\mathcal{L}\{t^{i}g_{I }\left( t \right)\}\left( s \right)$ and the Laplace identity $t^{i}g_{I }\left( t \right)\longleftrightarrow \left( -1 \right)^{i}\frac{\mathrm{d}^{n}}{\mathrm{d}^{n}s}\mathcal{L}_{I }\{g_{I }\left( t \right)\}\left( s \right)$. 
The Laplace transform $\mathcal{L}_{I_{l}}\left( s \right)$ is obtained by means of Proposition~\ref{LaplaceTransform}, while $\mathcal{L}_{I_{ \etv_{\mathrm{t}}}}$ and    $\mathcal{L}_{I_{ \etv_{\mathrm{r}}}}\!\left(s \right)$ are provided by Lemma~\ref{LaplaceTransformGamma}. Substitution of~\eqref{coverage3} into~\ref{coverage_definition1} concludes the proof.              
\section{Proof of Proposition~\ref{LaplaceTransform}}\label{LaplaceTransformproof}
Having defined $g_{l}$, accounting for  the interference channel coefficient, which has  identical distribution for all $l$, and for the transmit channel impairments from other BSs, the Laplace transform of the interference part $\mathcal{L}_{I_{l}}\left( s \right)$ can be derived as
\begin{align}
 &\!\!\mathcal{L}_{I_{l}}\left(  s \right)=\mathbb{E}_{I_{l}}\left[  e^{-s {I_{l}}}\right]=\mathbb{E}_{I_{l}}\left[  e^{-s \sum_{l\in \Phi_{B}\backslash x} p_{l}g_{l} y^{-\al}}\right]\nn\\
  &\!\mathop = \mathbb{E}_{\Phi_{B},g_{l}}\left[\prod_{l \in \Phi_{B}\backslash x} e^{-s p_{l} g_{l} y^{-a}} \right] \label{laplace 2}\\
&\!\mathop = \mathbb{E}_{\Phi_{B}}\left[\prod_{l \in \Phi_{B}\backslash x}  \mathcal{L}_{g_{l}}\left( s p_{l} y^{-a} \right)  \right]\label{laplace 3}\\
  &\!\mathop =   \mathrm{exp}\left( -{\lambda_{B}}\int_{\mathbb{R}^{2}} \left( 1-\mathcal{L}_{g_{l}}\left( s p_{l} y^{-a} \right)   \right)\mathrm{d}y\right)\label{laplace 4}\\
&\! \!\mathop =\mathrm{exp}\!\Bigg(\Bigg.\!\!\!-\!2 \pi \lambda_{B}\!\!\!\int_{0}^{\infty}  \!\!\bigg(\bigg. \frac{ \sum_{m=1}^{K}\binom{K}{m} \left( {s}p_{l}  r^{-a} \right)^{m}}{\left( 1\!+\!{s}p_{l}  r^{-a} \right)^{K}}r\mathrm{d}r\Bigg)\!\!\label{laplace5} \\
&\!\!\! \!\mathop =\mathrm{exp}\!\Bigg(\Bigg.\!\!-\frac{2 \pi \lambda_{B} {p_{l}}^{\frac{2}{a}}{s}^{\frac{2}{a}}}{\al}\!\!\sum_{m=1}^{K}\!\!\binom{\!K\!}{\!m\!}\mathrm{B}\!\left(\! K\!-\!m\!+\!\frac{2}{a},m\!-\!\frac{2}{a} \right)\!\!\!\!\Bigg),\!\!\!\nn
\end{align}
where  $\mathrm{B}\left( K-m+\frac{2}{a},m-\frac{2}{a} \right)$ is the Beta function defined in~\cite[Eq.~(8.380.1)]{Gradshteyn2007}.
Note that~\eqref{laplace 2} comes from the independence among the  locations of the BSs, while~\eqref{laplace 2} holds due to the independence between the spatial and the fading distributions. Using the property of PGFL of PPP~\cite{Chiu2013}, we obtain~\eqref{laplace 4}, and the next step follows by substituting the Laplace transform of $g_{l}$, obtained in Lemma $2$. Application of the Binomial theorem in~\eqref{laplace 4}, and conversion of the Cartesian coordinates to polar coordinates results to~\eqref{laplace5}. The proof is concluded with the calculation of the integral. Specifically, it is 
obtained after substitution of $\left( 1+r^{-\al} \right)^{-1}\rightarrow t$, many algebraic manipulations, and the use of~\cite[Eq.~(8.380.1)]{Gradshteyn2007}.
\end{appendices}

\bibliographystyle{IEEEtran}

\bibliography{mybib}
\end{document}